\documentclass{amsart} 

\usepackage{graphicx}
\usepackage{amsthm}
\usepackage{amsmath}
\usepackage{amssymb}
\usepackage{mathrsfs} 
\usepackage{url}
\usepackage{enumerate}
\usepackage{lscape}
\usepackage[usenames,dvipsnames]{color}
\usepackage{multicol}
\usepackage{multirow}
\usepackage[section]{placeins}
\usepackage{natbib}
\usepackage{tikz}

\newcommand{\R}{\mathbb{R}}
\newcommand{\C}{\mathbb{C}}
\newcommand{\Cp}{\mathcal C}

\newtheorem{theorem}{Theorem}[section]
\newtheorem{lemma}[theorem]{Lemma}
\newtheorem{definition}[theorem]{Definition}

\def\cprime{$'$}

\begin{document}

\title{Implicit dose-response curves}

\author{Mercedes P\'erez Mill\'an and Alicia Dickenstein}
\address{
    MPM and AD: Dto.\ de Matem\'atica, FCEN, Universidad de Buenos Aires, 
    Ciudad Universitaria, Pab.\ I, C1428EGA Buenos Aires, Argentina.  MPM: 
    Dto. de Ciencias Exactas, CBC, Universidad de Buenos Aires, Ramos 
    Mej\'{i}a 841, C1405CAE Buenos Aires, Argentina. AD: 
    IMAS - CONICET, Ciudad Universitaria, 
Pab.\ I, C1428EGA Buenos Aires, Argentina}
\email{mpmillan@dm.uba.ar, alidick@dm.uba.ar}
\thanks{The authors would like to thank the anonymous reviewers for their valuable suggestions. This work was 
  partially supported by UBACYT  20020100100242, CONICET PIP 11220110100580 and ANPCyT 2008-0902, Argentina. }



\begin{abstract}
We develop tools from computational algebraic geometry for the study of steady state features of
autonomous polynomial dynamical systems via elimination of variables.  In particular, 
we obtain nontrivial bounds for the steady state concentration of a given species in biochemical 
reaction networks with mass-action kinetics. This species is understood as the output of the network 
and we thus bound the \emph{maximal response} of the system.  The improved bounds
give smaller starting boxes to launch numerical methods. We apply our results to the sequential
enzymatic network studied in \citep{markevich} to find nontrivial upper bounds for the different
substrate concentrations at steady state.

Our approach does not require any simulation, analytical expression to describe the output in terms 
of the input, or the absence of multistationarity. Instead, we show how to extract information from 
effectively computable implicit dose-response curves, with the  use of resultants and discriminants. 
We moreover illustrate in the application to an enzymatic network, the relation between the {exact}
implicit dose-response curve we obtain symbolically and the standard hysteresis diagram provided by
a numerical ode solver.

The setting and tools we propose could yield many other results adapted to any autonomous polynomial 
dynamical system, beyond those where it is possible to get explicit expressions.

  
\end{abstract}

\maketitle

\section{Introduction}

Consider an autonomous polynomial dynamical system
\begin{equation}\label{eq:dx}
\dfrac{d}{dt}x(t)=f(x(t)) 
\end{equation}
where $x=(x_1,\dots,x_s)$ and $t$ are real variables, and each coordinate $f_i$ 
is a polynomial in $x_1, \dots, x_s$ with real coefficients. The \emph{steady states} 
of~\eqref{eq:dx} are thus the  real zeros of the algebraic variety defined by 
$f_1(x) = \dots =f_s(x)= 0$. An important example of these systems are chemical 
reaction networks with {\em mass-action kinetics}, which
have been extensively studied on a mathematical basis since the foundational
work by Feinberg \citep{feinberg}, Horn and Jackson \citep{hj72} and Vol{\cprime}pert \citep{vh85}. 
In this case, 
$x_1,x_2,\ldots,x_s$ represent species concentrations, considered as functions of 
time $t$ and the meaningful steady states are those with nonnegative coordinates. 
We will mainly use the terminology of chemical reaction networks throughout 
and consider nonnegative $x_i$.

Any  linear relation (with real coefficients) among the polynomials $f_1, \dots, f_s$ 
defines a conservation relation of the form 
\begin{equation}\label{eq:L}
 L(x) \, = \,  \ell(x) - b \, = \, 0,
\end{equation}
where $\ell$ is a homogeneous linear form in the variables $x_1, \dots, x_s$ and the 
constant $b = \ell(x^*) \in \mathbb{R}$ 
is determined by the initial values $x^*=x(0)$ of the system.

\

\begin{definition}  \label{def:trivial} 
We say that $b> 0$ is a  \emph{trivial upper bound} for the $i$th species if there exists a conservation relation  
$a_1 x_1+ a_2 x_2+ \dots +x_i + \dots + a_s x_s - b = 0$ with  all $a_j\geq 0$.
\end{definition}

\
In the particular important case of conservative networks, there are trivial upper 
bounds for the concentrations of all the species.
Note that in the conditions of Definition~\ref{def:trivial}, $b$ is an upper bound for the 
concentration of $x_i$ \emph{along the whole trajectory} in $\R_{\ge 0}^s$.
Our main goal is to improve these bounds for steady state 
concentrations of specific species of the system (that we will call {\it output}).
It is important to notice that, in general, there is no analytical expression to describe these concentrations 
and there could be multistationarity, which makes finding these bounds a difficult task. 

In the special bacterial EnvZ/OmpR osmolarity regulator, algebraic methods are used in \cite{kpmddg} 
to detect the existence of robust upper bounds at steady state, i.e., bounds that depend only on the 
reaction constants and not on the initial conditions or the total concentration of the species. 
Multistationarity in enzymatic networks has been studied with geometric and algebraic tools for example
in~\cite{fw12,fhc13,pdsc12,ws08}. A particular case of our approach has been studied in \cite{fkaw12} 
for signaling cascades with $n$ layers and one 
post-trans\-la\-tional modification cycle at each layer. A nontrivial bound for the maximal response 
of the modified substrate in the $n$-th layer can be read from a polynomial involving its concentration and 
the total amount of the first modification enzyme, which has degree
one in this second variable. This is the simplest case in our analysis, which is then reduced 
to studying the zeros of the leading coefficient. The authors  also present a deeper study of the bounds  by 
tracing back the values of the modified substrate in the $n$-th layer  which can be completed to a positive steady 
state of the whole system.

We consider for instance the steady state concentration of $x_1$ as
our {\it output} and the constant term  $c$ of
a particular conservation relation~\eqref{eq:L} as our 
{\it input}. In the chemical reaction network 
setting, $c$ usually stands for a total
concentration. We will find with methods of computational algebraic geometry 
--under natural hypotheses-- an implicit polynomial relation 
$p(c,x_1)=0$ between
the values of $x_1$ at steady state and $c$.
Note that in case of multistationarity, there will be several
$x_1$ satisfying this equation for the same value of the input $c$.
Assuming there is a trivial upper bound $b$, one can consider $c$ as the constant term
of a conservation relation linearly independent of the one giving $b$. If one is able to plot
the curve ${\mathcal C} = \{(c,x_1) ~|~ p(c,x_1)=0\}$, then an upper bound for the values of $x_1$ at steady state
can be read from this plotting. However, an implicit plot has in general bad quality and is
inaccurate. Instead, we appeal to the properties of resultants and discriminants to preview
a ``box'' containing the intersection of $\mathcal C$ with the first orthant in the plane
$(c,x_1)$. In fact, these tools are usually applied to produce the approximate
implicit plotting. The improved bounds give smaller starting boxes to launch numerical computations.
We will call $\mathcal C$ an {\em implicit dose-response curve}.
These implicit dose-response curves can also be used --via implicit dif\-feren\-tiation-- to study the
{\em sensitivities} of the local variation of $x_1$ around $c^*$ as a function of $c$ when $p(c^*,x_1^*)=0, 
\frac {\partial p} {\partial x_1}(c^*, x_1^*) \neq 0$, without an explicit
expression for the local function $x_1=x_1(c)$ in a neighborhood of $(c^*,x_1^*)$ in $\mathcal C$.

The approach we propose could yield many similar results.
As an application, we consider the mass-action system 
in \cite{markevich}
for the sequential double-phospho\-rylation enzymatic mechanism, 
which can give rise to multistationarity:

 \begin{tiny}
\begin{equation}\label{eq:red} 
\begin{split}
    \text{M} + \text{MAPKK}
    \underset{k_{-1}}{\overset{k_1}{\rightleftarrows}}
    \text{M-MAPKK} 
    \overset{k_2}{\rightarrow}
    \text{Mp} + \text{MAPKK}
     \underset{k_{-3}}{\overset{k_3}{\rightleftarrows}}
    \text{Mp-MAPKK}
    \overset{k_4}{\rightarrow}
    \text{Mpp} + \text{MAPKK} \\
    \text{Mpp} + \text{MKP3}
    \underset{h_{-1}}{\overset{h_1}{\rightleftarrows}}
    \text{Mpp-MKP3}
    \overset{h_2}{\rightarrow}
    \text{Mp-MKP3}
    \underset{h_{-3}}{\overset{h_3}{\rightleftarrows}}
    \text{Mp + MKP3}
    \underset{h_{-4}}{\overset{h_4}{\rightleftarrows}}
    \text{Mp-MKP3}^*
    \overset{h_5}{\rightarrow}
    \text{M-MKP3}
    \underset{h_{-6}}{\overset{h_6}{\rightleftarrows}}
    \text{M} + \text{MKP3}\\
\end{split}
\end{equation}
\end{tiny}
We feature the system in the form~\eqref{eq:dx} in \S~\ref{ssec:eq}.
There are eleven variables given by the concentrations of the eleven chemical species:
the unphosphorylated substrate M, the singly phosphorylated substrate Mp and the
doubly phosphorylated substrate Mpp, the two enzymes (the kinase MAPKK and the
phosphatase MKP3) plus the six intermediate species.
There are three independent conservation relations (also translated to $x_i$ variables
in \S~\ref{ssec:eq}):
{\tiny
\begin{align*}
[\text{M-MAPKK}]+[\text{Mp-MAPKK}]+[\text{MAPKK}]  - \text{MAPKK}_{tot} &  = 0,\\
[\text{Mpp-MKP3}]+[\text{Mp-MKP3}]+[\text{Mp-MKP3}^*]+[\text{M-MKP3}]+[\text{MKP3}]  - \text{MKP3}_{tot} &=0,\\
[\text{M}]+[\text{Mp}]+[\text{Mpp}]+[\text{M-MAPKK}]+[\text{Mp-MAPKK}]+[\text{Mpp-MKP3}] & +\\
+ [\text{Mp-MKP3}]+[\text{Mp-MKP3}^*]+[\text{M-MKP3}]  - \text{M}_{tot}& = 0.
\end{align*}
}
The usual output of this network is the concentration $x_1=$[Mpp] of the doubly phosphorylated 
substrate Mpp. Consider as an {\it input} of this network the total amount $c=$MAPKK$_{tot}$ 
related to the kinase MAPKK.  We easily deduce from the third conservation relation that 
$b=\text{M}_{tot}$ is a trivial upper bound for [Mpp] along the whole trajectory.
We find nontrivial bounds for this species at steady state, which are also 
{\em independent} of the input value. 
Our analysis shows how to ``regulate'' the parameters 
of the system in a 
more explicit way than simply running a simulation of the complete system.

We give in Section~\ref{sec:results} sufficient conditions to 
find nontrivial upper bounds by using tools from computational algebraic geometry, 
in particular variable elimination and 
the notion of discriminant \citep{GKZ}. Our main theoretical results are 
summarized in Theorem~\ref{thm:main}.
We then apply in Section~\ref{sec:4} our results to show nontrivial bounds for the 
concentration of the doubly-phosphorylated substrate in the 
sequential double-phosphorylation system presented in \cite{markevich},
showing how to exploit the implicit dependencies obtained with a computer algebra system.
We moreover point out the relation of the implicit dose-response curve $\mathcal C$
with the hysteresis graphs interpolated by numerical ode solvers. An appendix contains the
proofs of the theoretical results.

\section{Methods and results}\label{sec:results}

Our main result is Theorem~\ref{thm:main}, which can be seen as a sample statement, in the following sense:
there are
many other similar results which could be proved with the tools we present, 
adapted to different families of autonomous polynomial dynamical systems.

We assume the dimension $r$ of the space of the homogeneous linear forms defining
conservation relations is positive, and take a basis $\ell_1, \ell_2, \dots,\ell_r$ of this subspace. 
In the context of chemical reaction systems,
the linear equations defining the so called stoichiometric subspace give in general all
the conservation relations \citep{feho77}.
We will consider the constant term $c=b_1$ of $\ell_1$ 
as our {\em input} and one of the $x$-variables, say $x_1$,
as our {\em output}.

We will look for \emph{steady state invariants} which are 
\emph{polynomial consequences} of the equations 
 \begin{equation} \label{eq:efes}
f_1=f_2= \dots =f_s=\ell_1-c=\ell_2 - b_2= \dots=\ell_r - b_r=0,
\end{equation}
that we will use to detect properties of the concentrations at steady state.
So, we will not only look for linear combinations of our equations with real number
coefficients, but also
with real polynomial coefficients. This is made precise in the definition of the ideal 
 $I$ generated by 
$f_1, f_2, \dots, f_s, \ell_1 - c, \ell_2 - b_2, \dots, \ell_r - b_r$
in the polynomial ring ${\mathbb R}[c,x_1\dots, x_s]$:
\[
I = \left\lbrace \overset{s}{\underset{j=1}{\displaystyle \sum}} g_jf_j+
g_{s+1}(\ell_1 - c)+\overset{r}{\underset{k=2}{\displaystyle \sum}} g_{s+k}(\ell_k- b_k)\right\rbrace,
\]
where  $g_1, \dots, g_{s+r}$ are polynomials in the variables
$c, x_1, \dots, x_s$. For a chemical reaction system, the real nonnegative common zero set 
of all the polynomials in $I$ coincides with the 
steady states in the \emph{stoichiometric compatibility class} determined by $c, b_2, \dots, b_r$.
We refer the reader to the nice book \cite{IVA} for a basic introduction to the concepts and tools
from computational algebraic geometry we use. The proofs of our results can be found in the Appendix.

\begin{lemma}\label{lem:p}
With the previous notations, assume that system~\eqref{eq:efes}
has finitely many complex solutions $(x_1, \dots,x_s)$ for any value of $c$.
Then, it is possible to construct a nonzero polynomial 
$p= p(c,x_1)$  in $I$ only depending on $x_1$ and $c$ and with positive degree in $x_1$.
\end{lemma}
Such a polynomial $p$ gives an implicit relation between $x_1$ and $c$ at steady state. 
It can be computed effectively by standard elimination techniques from 
computational algebraic geometry. The hypothesis of finitely many complex solutions 
does hold in most biological examples and it is 
always {\emph{assumed tacitly}}. 
For readers with enough algebraic geometry background, we remark that in fact, 
for Lemma~\ref{lem:p} to hold, it is enough to ask the two conditions
we state in the following paragraph.

Note that we can choose $s-r$ linearly independent $f_i$'s, say $f_1,\dots, f_{s-r}$, 
and so $I$ can be generated by the $s$ polynomials 
$f_1,\dots, f_{s-r}, \ell_1-c,\dots,\ell_r-b_r$ in $s+1$ variables $c,x_1,\dots, x_s$,
as $f_{s-r+1}, \dots, f_s$ are $\R$-linear combinations of $f_1, \dots, f_{s-r}$.
So,  it holds that the dimension
of the ideal $I$ equals one for general coefficients. This is the first condition.
The second natural condition requires that there is no nonzero polynomial only
depending on $c$ lying in $I$. This means that system~\eqref{eq:efes} has a solution for
infinitely many values of $c$, which also
holds in general.

From a polynomial $p=p(c,x_1)$ as in Lemma~\ref{lem:p}, 
we can establish bounds for the steady state concentration of $x_1$. As a first step, 
for any given $c=c^*$, the $x_1$ coordinate of any steady state is a root of the univariate polynomial $p(c^*,x_1)$, 
which can be approximated or bounded in terms of its coefficients.  Note that there could be 
multistationarity for this particular value $c^*$ and we can estimate \emph{all} possible values of 
$x_1$ for any given nonnegative initial condition.

In what follows, we will present a way of getting bounds which hold for \emph{any}  meaningful value
of the input $c$. It might happen that $p$ does not depend on $c$. In this exceptional case, the $x_1$
coordinates of any steady state can only equal the (finite number of) nonnegative real
roots of $p=p(x_1)$, for any $c$. In what follows, we assume that the degree $n$ of $p$ in $c$  
is positive and write 
\begin{equation}\label{eq:p}
 p=\sum_{i=0}^n p_i(x_1)c^i, \quad p_n \neq 0.
\end{equation}
In order to understand the intersection of the first orthant with the implicit dose-response curve 
$\Cp= \{(c,x_1)~|~p(c,x_1)=0\}$,
we will use the notions of resultant and discriminant~\citep{GKZ}.
The resultant 
\vspace{-.1cm} 
\begin{equation} \label{eq:Rn}
R_{n}:= \, {\rm Res}_{n, n-1} \left(p,\frac{\partial p}{\partial c},c\right) \; \in \R[x_1], 
\end{equation}
of $p$ and $ \frac{\partial p}{\partial c}$, thought of as polynomials  in $\mathbb{R}[x_1][c]$ 
of degree $n$ and $n-1$, respectively, is a  
polynomial in the variable $x_1$ which characterizes the existence of common roots of $p(c,x_1^*)$ and its derivative
with respect to $c$,  for values $x_1^*$ of $x_1$ for which the degree of $p(c,x_1^*)$ in the variable $c$ is $n$.

Take any  
fixed $x_1^*$ such that $p_n(x_1^*) \not =0$, so that the specialized polynomial $p(c,x_1^*)$ 
has degree $n$ in $c$. The discriminant of $p(c,x_1^*)$ (with respect to $c$) depends polynomially
on $x_1^*$ and defines a polynomial $D_n  \in \R[x_1]$. By definition, $D_n(x_1^*)=0$
 if and only if there is a (complex) value of $c$ for which 
$p(c,x_1^*)=\frac{\partial p}{\partial c}(c,x_1^*) =0$.  When there exists a real solution $c^*$,  
this condition is equivalent to the fact that
the curve $\Cp$ has a tangent which is parallel to the $c$-axis at the point $(c^*,x_1^*)$.
On the other side, if the line $x_1=\alpha$ is an asymptote of the curve $\Cp$, 
that is, if there exists a sequence
$(c^{(m)}, x_1^{(m)}) \in {\mathcal C}$ with $c^{(m)} \to \infty$ and $x_1^{(m)} \to \alpha$,  then $p_n(\alpha)=0$.

We have the following characterization of the zeros of the resultant~\eqref{eq:Rn} \citep[see][chap.~12~\S~1]{GKZ}.

\begin{lemma} \label{lem:res}
The zeros of $R_n$ in the variable $x_1$ are given by the union of the roots of the leading coefficient
$p_n$ and the roots of the discriminant $D_n$ of $p$ as a polynomial in the variable $c$. 
\end{lemma}
 The resultant $R_{n}$ can be computed as the determinant of 
the corresponding $(2n-1) \times (2n-1)$ Sylvester matrix  (or by smaller matrices, involving the Bezoutian).

 The general framework where we could use $p$ to get nontrivial bounds for the steady
 state values of $x_1$ is the following.
 We assume that system~\eqref{eq:dx}  has a 
nonnegative conservation relation  $L =\ell-b$ as in~\eqref{eq:L},  
in which $x_1$ appears with nonzero coefficient 
and all the other coefficients in $\ell$ are nonnegative. This gives 
a trivial bound for the steady 
state value of $x_1$. We furthermore assume that $r \ge 2$ and $\ell_1$
is linearly independent from $\ell$.
We can obtain bounds for the values of $x_1$ (independent of $c$), once
the values $b_2, \dots, b_r$ of the conservation relations associated to $\ell_2, \dots, \ell_r$
have been fixed.

We give now our main result. To state it, we introduce the following notations.
For any fixed $\gamma \in {\mathbb R}$, we will denote by $\Cp_\gamma$ the intersection
of $\Cp$ with the horizontal line $\{x_1=\gamma\}$:
\begin{equation}\label{eq:Cgamma}
 \Cp_\gamma:=\{c\in \mathbb{R} ~|~ p(c,\gamma)=0\},
\end{equation}
and we denote by $J$ the image  
\[J:=\ell_1(\mathbb{R}_{\geq 0}^s)\] 
 of the nonnegative orthant by the linear form $\ell_1$.
Note that if the signs of all coefficients in $\ell_1$
 are the same, we can assume they are all nonnegative and then 
 $J=[0,+\infty)$; otherwise, $J=\mathbb{R}$.

\begin{theorem}\label{thm:main}
Consider $p=p(c,x_1) \in I$ with positive degree $n$ in $c$ such that the resultant $R_n \not\equiv 0$. 
 Let $\{\alpha_1, \alpha_2, \dots, \alpha_m\}$ be the set of real zeros of $R_{n}$, 
 with $\alpha_1 > \dots > \alpha_m$.
  If for some index $k \in \{1,\dots,m\}$ there exist $\beta_1, \dots, \beta_k \in \mathbb{R}$ with
 $$\beta_1>\alpha_1 > \beta_2 > \alpha_2 > \dots > \beta_k > \alpha_k$$ 
 such that for all $1\leq i \leq k$,
 $\Cp_{\beta_i} = \emptyset$  and  $\Cp_{\alpha_i} \cap J = \emptyset$,  then $x_1 < \alpha_k$  
 at any steady state. In other words, $\alpha_k$ is an upper bound for $x_1$ at steady state.

 Moreover, let $\alpha$ denote the biggest positive real root of $p_n$ and assume that $\alpha < \alpha_k$.
Assume  $\Cp_{\gamma} \cap J = \emptyset$ for all roots
 $\gamma$ of $R_{n}$ in the interval $[\alpha, \alpha_k]$. 
 In case $J =[0,+\infty)$, assume also that 
 the univariate polynomial $p(0,x_1)$ does not have any positive real roots bigger than $\alpha$.
 Then,  $\alpha$ is a more precise upper bound 
 for $x_1$ at steady state.
\end{theorem}

 We illustrate  in Section~\ref{sec:4} the improvement in the maximal response given by 
 Theorem~\ref{thm:main} in the interesting example of the sequential phosphorylation of \cite{markevich}. 
 Considering the polynomial $p$ in that section, we depict  in Figure~\ref{fig:tang} (a) 
 the curve $\Cp$ and the values of $\alpha_1, \alpha_2, \alpha_3$  (detailed in
 \S~\ref{ssec:x1c}), together with the trivial bound $500$. We also show in the adjacent image 
 (b) that the occurrence of $\alpha_2$ is due to a horizontal tangency at a point with negative
 value of $c$. For more details, see Figures~\ref{fig1},\ref{fig1234}.
\begin{figure}[ht]
\begin{tabular}{cc}
  (a) \includegraphics[scale=.23,trim=2mm 5mm 2mm 7mm, clip=true]{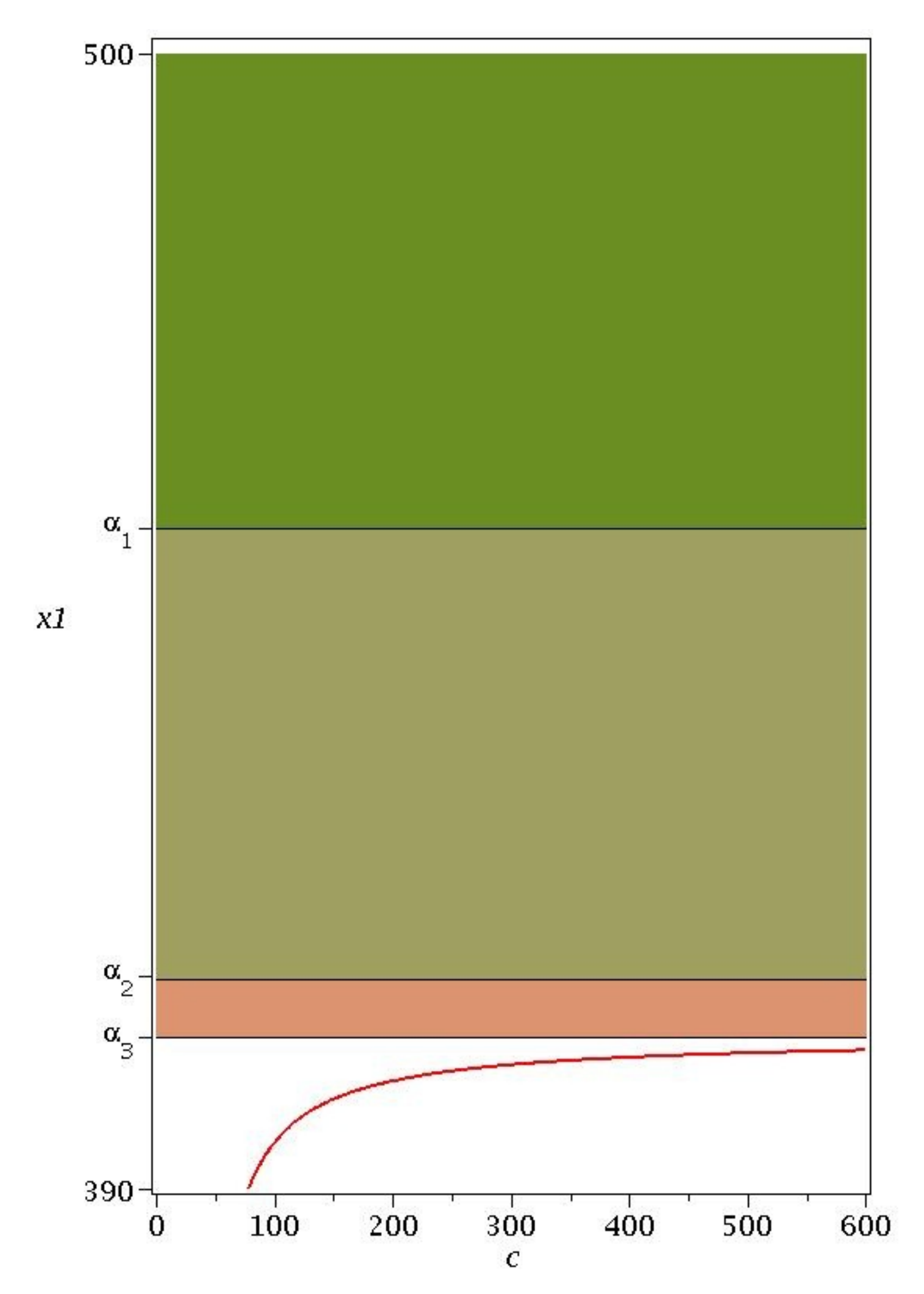} &
  (b)\includegraphics[scale=.26,trim=2mm 5cm 2mm 5cm, clip=true]{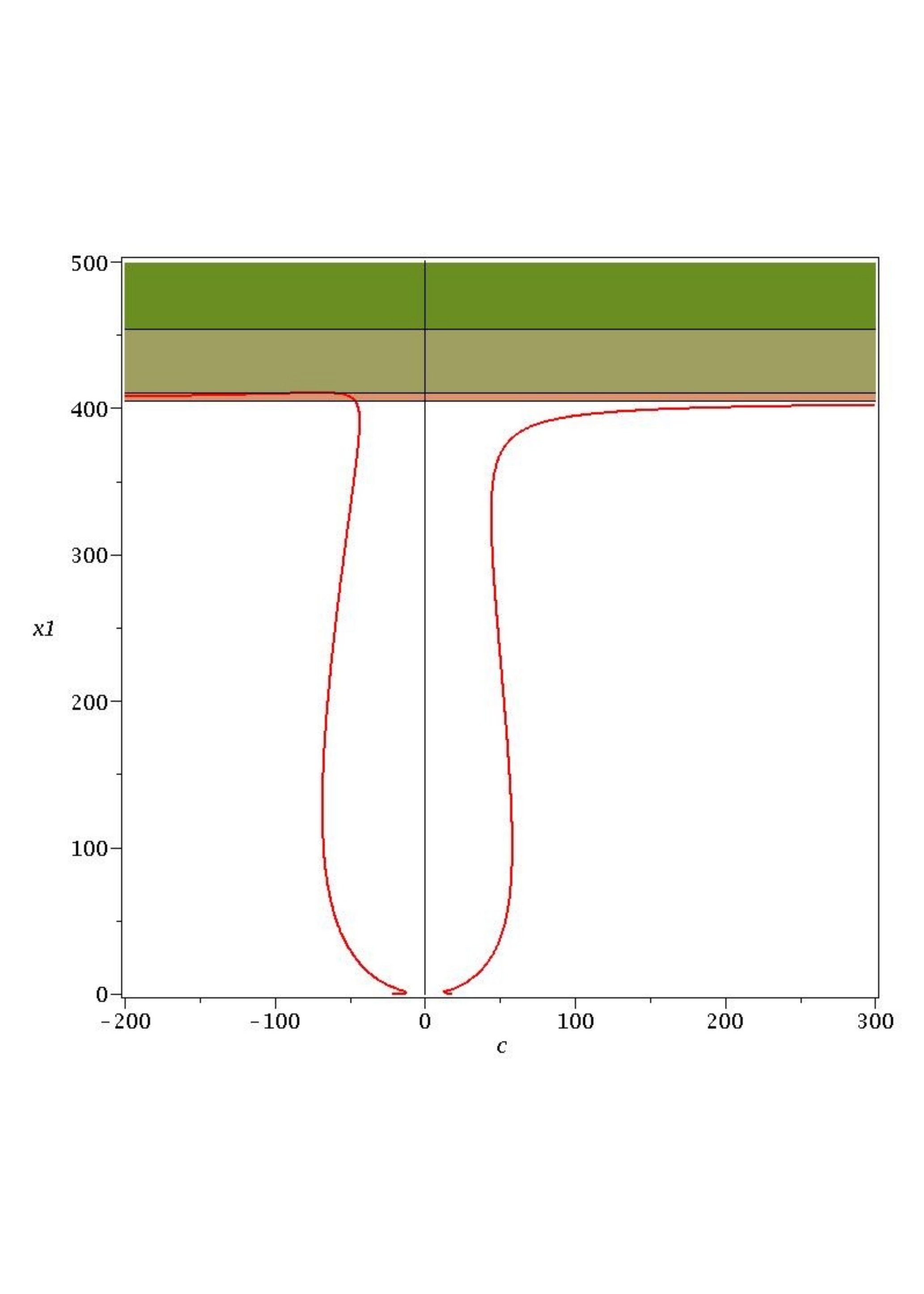}
\end{tabular}
\caption{Plot of the implicit curve $\Cp$ with Maple for the sequential phosphorylation in \cite{markevich}. 
(a): The bounds $500, \alpha_1, \alpha_2,\alpha_3$ for $c > 0, x_1 > 390$, where the intervals 
$[\alpha_{i+1},\alpha_i]$ have different colors. 
(b): The picture  for $-200 < c < 300, x_1 > 0$.}
\label{fig:tang}
\end{figure}

The first part of Theorem~\ref{thm:main} is based on the following well known result, which
follows from the Implicit Function Theorem (IFT). As we haven't found any good reference for its proof, 
we sketch it in the Appendix for the convenience of the reader.

\begin{lemma}\label{lem:IFT}
Let $p=p(c,x_1) \in I$ with positive degree $n$ in $c$ and for any $\beta$ consider the set
$\Cp_{\beta}$ defined in~\eqref{eq:Cgamma}.
Then, the cardinality $\#\Cp_{\beta}$ of $\Cp_{\beta}$ is the same for all $\beta$ 
in a connected component $\Omega$ of the complement of the zeros of the 
resultant $R_n$ in $\R$.
\end{lemma}

Under the hypotheses of Lemma~\ref{lem:p}, there exists a polynomial $p \in I$ with
positive degree in $x_1$. As we remarked before, unless $x_1$ takes only a finite number of
values, this polynomial will also have positive degree in $c$, which is required
in Theorem~\ref{thm:main}. Indeed, as also $R_n$ is required to be non identically zero,
if the degree of $p$ in $x_1$ is not positive, then $R_n$ would
be a nonzero constant. Therefore, $R_n$ would have no roots and the result is void.

We observe that there is no need to have the exact values of the roots of $R_n$ 
(which are in general impossible to get). It is
enough to find (small) intervals that isolate the roots (say, of radius $\delta$ 
around each $\alpha_j$) and then pick the values $\beta_j$ between the 
extreme points of these intervals.
The bound we get this way is slightly bigger (e.g. $\alpha_k +\delta$), but computable.
On the other side, in order to check the emptiness of $\Cp_{\alpha_i} \cap \ell_1(\mathbb{R}_{\geq 0}^s)$,
there are symbolic procedures available to determine
the number of real roots of zero dimensional ideals subject to real polynomial inequalities, 
for example the libraries for real
roots implemented in Singular~(\citeauthor{singular};\citealp{tobis}). Namely, if $\alpha_i$ is the unique root of $R_n$
in the rational interval $(\xi_1,\xi_2)$, then one needs to check that there are no real solutions $c$
satisfying the conditions
$$ R_n(x_1) = p(c,x_1)=0, \, \xi_1 < x_1 < \xi_2, \,  c \in J.$$
Notice also that the bounds in Theorem~\ref{thm:main} hold in principle for fixed values of $b_2, \dots, b_r$, but in theory
one could get (by a variant of Lemma~\ref{lem:p} under natural hypotheses) a polynomial $p$
depending on these parameters (and even on the rate constants). We exemplify this in \S~\ref{ssec:moving}.

Our methods can be adapted, besides mass-action kinetics systems, to standard modelings with 
autonomous rational dynamical systems, like
power law dynamics with integer exponents or Michaelis-Menten kinetics.

\section{Application to an enzymatic network}\label{sec:4}

In this section, we illustrate the use of 
 Theorem~\ref{thm:main} to find nontrivial bounds in example~\eqref{eq:red} from
\cite{markevich}, which models an enzymatic network with sequential phosphorylations and dephosphorylations. 
We also use this example to explain the need for the hypotheses and the scope of
Theorem~\ref{thm:main}. We moreover use the tools presented in Section~\ref{sec:results} to get 
a more detailed study of the system.

\subsection{The equations}\label{ssec:eq}
We name the species concentrations in network~\eqref{eq:red} by
{\small
$$\begin{array}{llll}
 x_1 \leftrightarrow [\text{Mpp}] , &x_4 \leftrightarrow [\text{M-MAPKK}] , &x_6 \leftrightarrow [\text{Mpp-MKP3}] , &x_{10} \leftrightarrow [\text{MAPKK}],\\
 x_2 \leftrightarrow [\text{Mp}] , &x_5 \leftrightarrow [\text{Mp-MAPKK}] , & x_7 \leftrightarrow [\text{Mp-MKP3}] , &x_{11} \leftrightarrow [\text{MKP3}].\\
 x_3 \leftrightarrow [\text{M}] , & & x_8 \leftrightarrow [\text{Mp-MKP3}^*] ,&\\
 & & x_9 \leftrightarrow [\text{M-MKP3}] , &\\
\end{array}$$}
Then, the differential equations of the system under mass--action kinetics are:
{\small
\begin{align*}
f_1 & = k_4x_5-h_1x_1x_{11}+h_{-1}x_6\\
f_2 & = k_2x_4-k_3x_2x_{10}+k_{-3}x_5+h_3x_7-(h_{-3}+h_4)x_2x_{11}+h_{-4}x_8\\
f_3 & = -k_1x_3x_{10}+k_{-1}x_4-h_{-6}x_3x_{11}+h_6x_9\\
f_4 & = k_1x_3x_{10}-(k_{-1}+k_2)x_4\\
f_5 & = k_3x_2x_{10}-(k_{-3}+k_4)x_5\\
f_6 & = h_1x_1x_{11}-(h_{-1}+h_2)x_6\\
f_7 & = h_2x_6-h_3x_7+h_{-3}x_2x_{11}\\
f_8 & = h_4x_2x_{11}-(h_{-4}+h_5)x_8\\
f_9 & = h_5x_8-h_6x_9+h_{-6}x_3x_{11}\\
f_{10} & = -k_1x_3x_{10}+(k_{-1}+k_2)x_4-k_3x_2x_{10}+(k_{-3}+k_4)x_5\\
f_{11} & = -h_1x_1x_{11}+h_{-1}x_6+h_3x_7-(h_{-3}+h_4)x_2x_{11}+h_{-4}x_8+h_6x_9-h_{-6}x_3x_{11},
\end{align*}}
and the conservation relations can be given as:
{\small
\begin{align*}
L_1=x_4+x_5+x_{10} -\text{MAPKK}_{tot} &=0\\
L_2=x_1+x_2+x_3+x_4+x_5+x_6+x_7+x_8+x_9  -\text{M}_{tot} &=0\\
L_3=x_6+x_7+x_8+x_9+x_{11} -\text{MKP3}_{tot} &=0.
\end{align*}}
We set the reaction constants as in the SI in \cite{markevich}:

\noindent $k_1= 0.02, k_{-1}=1, k_2=0.01, k_3=0.032, k_{-3}=1, k_4=15$, $h_1=0.045, h_{-1}=1, h_2=0.092, h_3=1$, 
 $h_{-3}=0.01,  h_4=0.01, h_{-4}=1, h_5=0.5,  h_6=0.086, h_{-6}=0.0011$,
and fix $\text{M}_{tot}=500, \text{MKP3}_{tot}=100$.
We let
\begin{align}\label{eq:elle1}
 \ell_1 &=x_4+x_5+x_{10},\\
 \nonumber \ell_2 &=x_1+x_2+x_3+x_4+x_5+x_6+x_7+x_8+x_9,\\
 \nonumber \ell_3 &=x_6+x_7+x_8+x_9+x_{11}.
\end{align}
Denote by $I$ the ideal generated by the polynomials $f_1, f_2, \dots, f_s, \ell_1-c, \ell_2-500, \ell_3-100$. 

\subsection{The implicit dose-response curve associated to $x_1$ and MAPKK$_{tot}$}\label{ssec:x1c}
We first take the output $x_1:=$[M$_{pp}$] and the input $c:=$MAPKK$_{tot}$. Note that the trivial 
bound along trajectories is equal to
$\text{M}_{tot}=500$. 

Via Gr\"obner basis elimination methods in \citeauthor{singular} we find that the intersection of $I$ with the ring of 
polynomials in the variables $x_1$ and $c$ is generated by the following  polynomial $p=p(c,x_1)=\sum_{i=0}^4 p_i(x_1) c^i$
with degree $n=4$  in  $c$, with coefficients:
\begin{small}
\begin{align*}
p_4 & = 259578228128346056201372100 x_1^4 - 91228131699664084594014546000 x_1^3 \\
 & - 5318853461888966748775026000000 x_1^2 - 107717641535472295661334000000000 x_1 \\
 & - 983693913810151954410000000000000,\\
p_3 & = -1279181837636260017061541940 x_1^5 + 217225713953041585784715122400 x_1^4 \\
 & + 111432561952880309835561787920000 x_1^3 + 4108996025231164151414890560000000 x_1^2 \\
 & + 32909012963892503562524400000000000 x_1, \\
p_2 & = 1651342827133987314483094029 x_1^6 - 57239961872970579411022490540 x_1^5 \\
 & - 172636108121018180634948973020000 x_1^4 - 29157440247951003530589295575600000 x_1^3 \\
 & - 655794481210925030267002164000000000 x_1^2 \\
 & - 3924727591361860067680350000000000000 x_1, \\
p_1 & = -23638737258912336217603357320 x_1^6 + 40121950932074520838137058397200 x_1^5 \\
 & - 10189010265838070554939750993840000 x_1^4 \\
 & - 458846180258284496202210449400000000 x_1^3 \\
 & - 3875235380408791071737337000000000000 x_1^2 \; \text{and}\\
p_0 & = 13225968047392416670218470096400 x_1^6 - 11693689998883687367816615216864000 x_1^5 \\
 & + 2446220546414268358687194986380000000 x_1^4 \\
 & + 67830374851435086233478373420000000000 x_1^3 \\
 & + 440552682490042857644978250000000000000 x_1^2.\\
\end{align*}
\end{small}
It is clear that one does not want to find this polynomial by hand. But once we have it, we can
extract interesting conclusions.

The resultant $R_{4}$ of $p$ and 
$\frac{\partial p}{\partial c}$, thought of as polynomials in $\mathbb{R}[x_1][c]$ of degrees $4$ and $3$, is a 
polynomial in $\mathbb{R}[x_1]$ of degree $38$ with big coefficients. 
$R_4$ has fourteen nonreal roots (of which eight are double roots), 
three negative roots (of which one is a double root), four positive real roots 
(of which two are double roots) and has $x_1=0$ as a root of multiplicity six.
The values of the positive roots are approximately: 
\[\alpha_1\approx 454.01, \; \alpha_2\approx 410.37, \; \alpha_3\approx 404.67, \, \text{and}\; \alpha_4\approx 312.56.\]
Following Theorem~\ref{thm:main}, we can choose for example $\beta_1= 470$, $\beta_2=420$, $\beta_3=405$ and $\beta_4=350$ which 
satisfy the inequalities $\beta_1>\alpha_1>\beta_2>\alpha_2>\beta_3>\alpha_3>\beta_4>\alpha_4$.
We find that $p(\beta_1,c)$ and $p(\beta_2,c)$ have no real roots. 
Since the nonzero coefficients of $\ell_1$ in~\eqref{eq:elle1} are positive, we have $J=[0,+\infty)$, 
and $\Cp_{\alpha_1}\cap J=\Cp_{\alpha_2}\cap J=\emptyset$. This makes $\alpha_2$ a 
nontrivial upper bound for $x_1$ at steady state by the first part of Theorem~\ref{thm:main}, which is sharper 
than the trivial bound $500$. 

The leading coefficient $p_4$ equals $3916521308700$ times the polynomial
{\small
\[(819x_1^2-308940x_1-9100000)(80925216157x_1^2+2085327062000x_1+27600573000000).
\]
}
As the second factor has no real roots, the real roots of $p_4$ are the roots of $819x_1^2-308940x_1-9100000$, which are, 
$\alpha_3$ and another one approximately equal to $-27.46$. 
As $J=[0,+\infty)$, we consider the zeros of $p(0,x_1)$, which are approximately 
$-13.52$, $-11.85$ and $0$. They are clearly less than $\alpha_3$, 
and $p(\alpha_2,c)=0$ for $c\approx -70.4$, which is negative. 
Then, by the second part of Theorem~\ref{thm:main}, 
$\alpha_3$ is a better nontrivial upper bound for $x_1$ at steady state (since $\alpha_3 < \alpha_2 < 500$).

\bigskip

If instead, we consider $x_3$ as output variable, 
by elimination in $I$ of all variables except for $x_3$ and $c$, we obtain a polynomial 
$q(x_3,c) \in I$ with degree $n=4$ in the variable $c$. Its leading coefficient $q_4$ is 
\[ 375791837967 x_3 (890177377727 x_3^2+81709195989640 x_3+3800250418891200).\]
Note that ${q}_4$ has no positive real roots, which makes us unable to apply the second 
part of Theorem~\ref{thm:main}. The resultant $R_{4}$  of $q$ and 
$\frac{\partial q}{\partial c}$, thought of as polynomials in $\mathbb{R}[x_3][c]$ of degrees $4$ and $3$, is a 
polynomial in $\mathbb{R}[x_3]$ of  degree $38$. $R_4$ has only one positive real root which is approximately  
$\alpha_1\approx 440.55$.
We can see that $q(450,c)$ has no real roots. This makes $\alpha_1$ a nontrivial upper bound for 
$x_3$ at steady state by the first part of Theorem~\ref{thm:main}.

\subsection{Depicting the implicit dose-response curve}\label{ssec:depict}

We depict in Figures~\ref{fig1} and~\ref{fig1234} the results we have obtained 
for the sequential dual phosphorylation--dephos\-pho\-ry\-la\-tion cycle from \cite{markevich} 
using the implicit dose-response curve ${\mathcal C} = \{(c,x_1) ~|~ p(c,x_1)=0\}$, plotted with \citeauthor{maple}. 
In Figure~\ref{fig1} we can see the curve 
$\mathcal{C}$ in the positive quadrant (the usual dose-response curve). This curve 
represents the relation between the input $\text{MAPKK}_{tot}$ ($c$) 
and the output Mpp ($x_1$) at steady state when both take positive values. The 
difference between the trivial and the nontrivial bounds is marked with color. 
In Figure~\ref{fig1234}(a) we can see that the nontrivial bound is not an upper bound for 
negative values of $c$, where $\alpha_2$ gives a slightly bigger upper bound. 
Figure~\ref{fig1234}(b) shows the curve $\mathcal{C}$ for 
negative values of $x_1$. 
\begin{figure}[ht]
\begin{center}
\includegraphics[scale=.29,trim=2mm 45mm 2mm 5cm, clip=true]{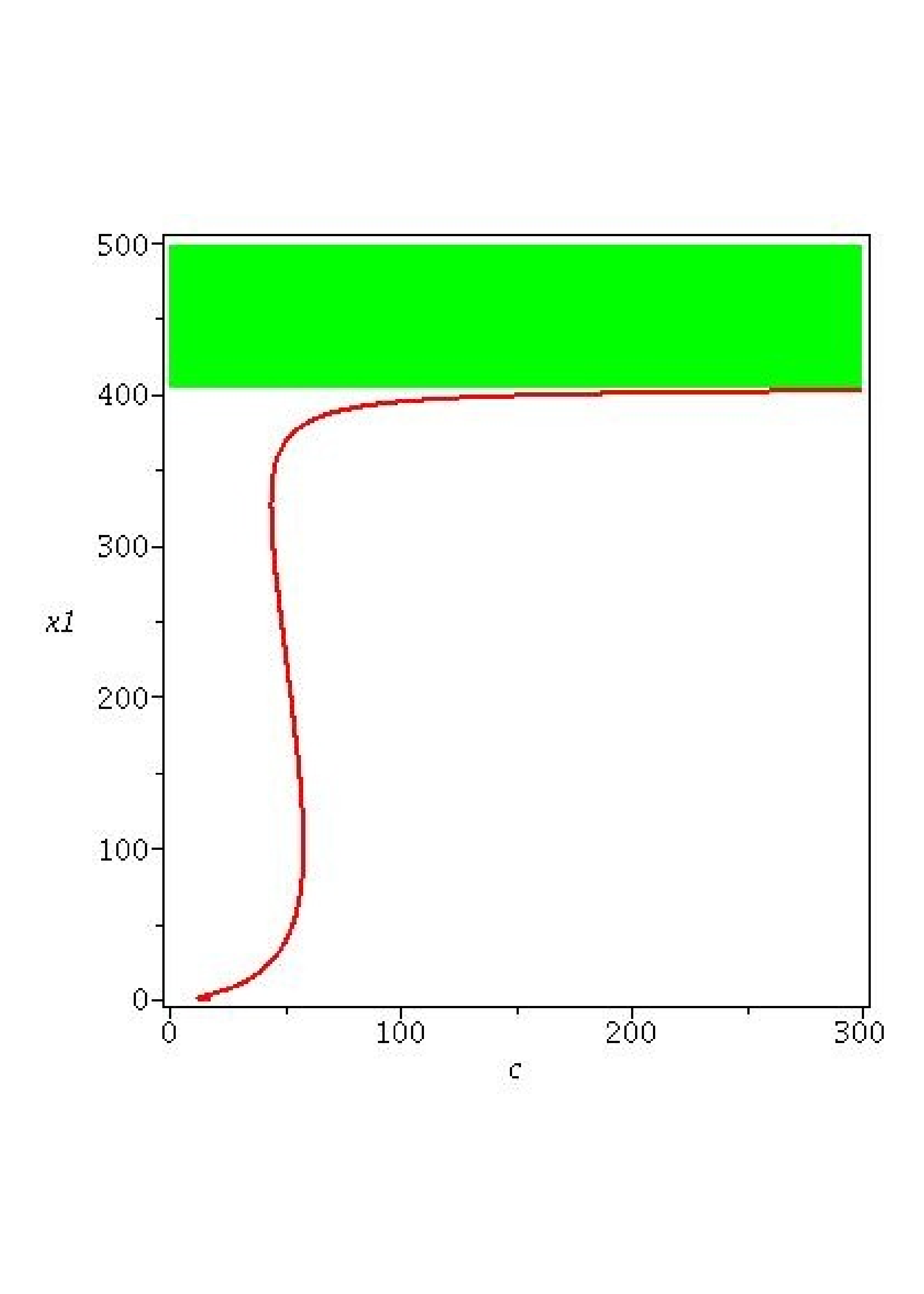}
\end{center}
\caption{The curve $\Cp$ in the positive quadrant. This curve 
represents the relation between the input $\text{MAPKK}_{tot}$ ($c$) 
and the output Mpp ($x_1$) at steady state when both take positive values. The 
difference between the trivial and the improved bound is marked with color.}
\label{fig1}
\end{figure}
\begin{figure}[ht]
\begin{tabular}{cc}
 (a)\includegraphics[scale=.28,trim=2mm 47mm 2mm 5cm, clip=true]{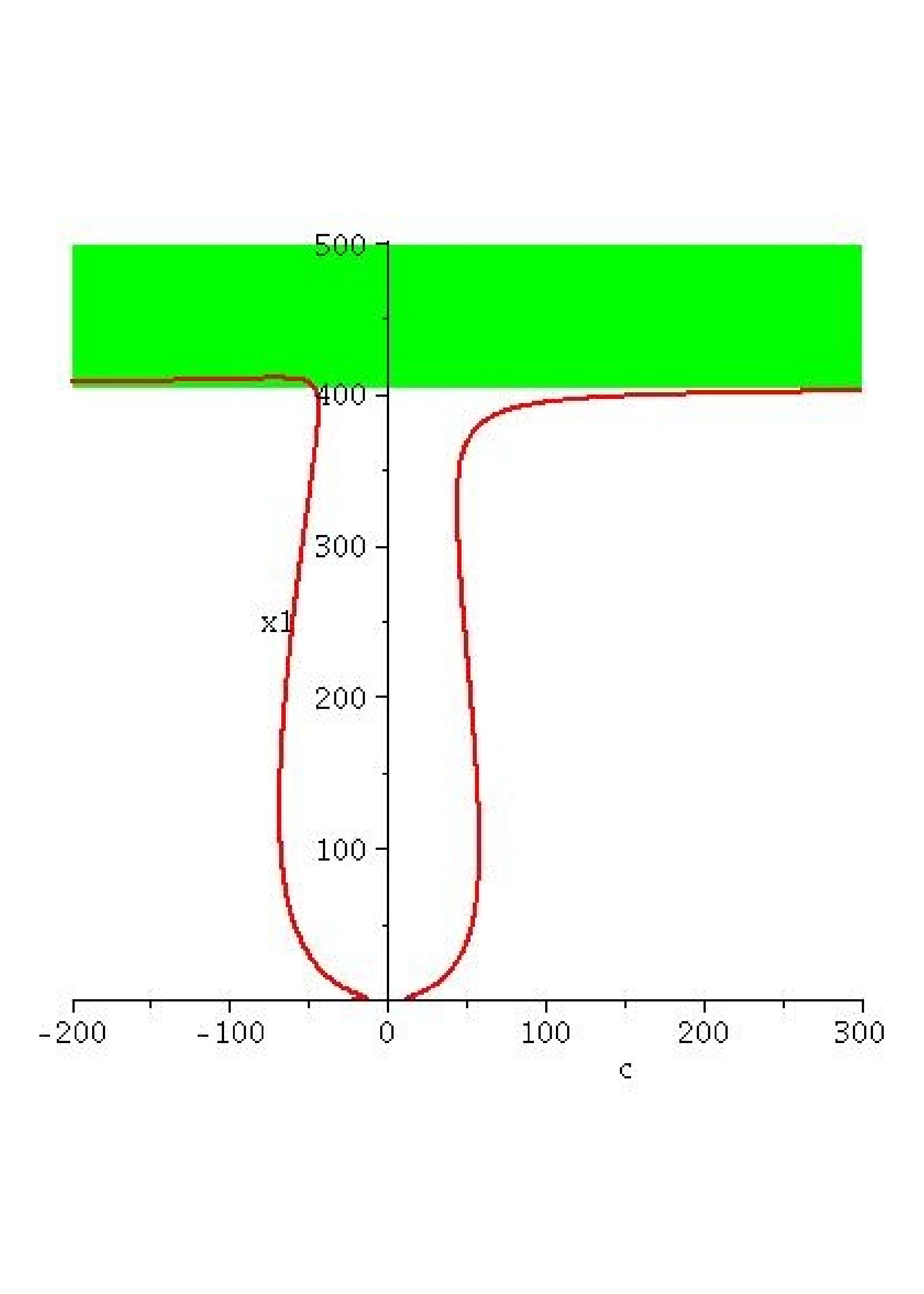}&
 (b)\includegraphics[scale=.22,trim=2mm 47mm 2mm 5cm, clip=true]{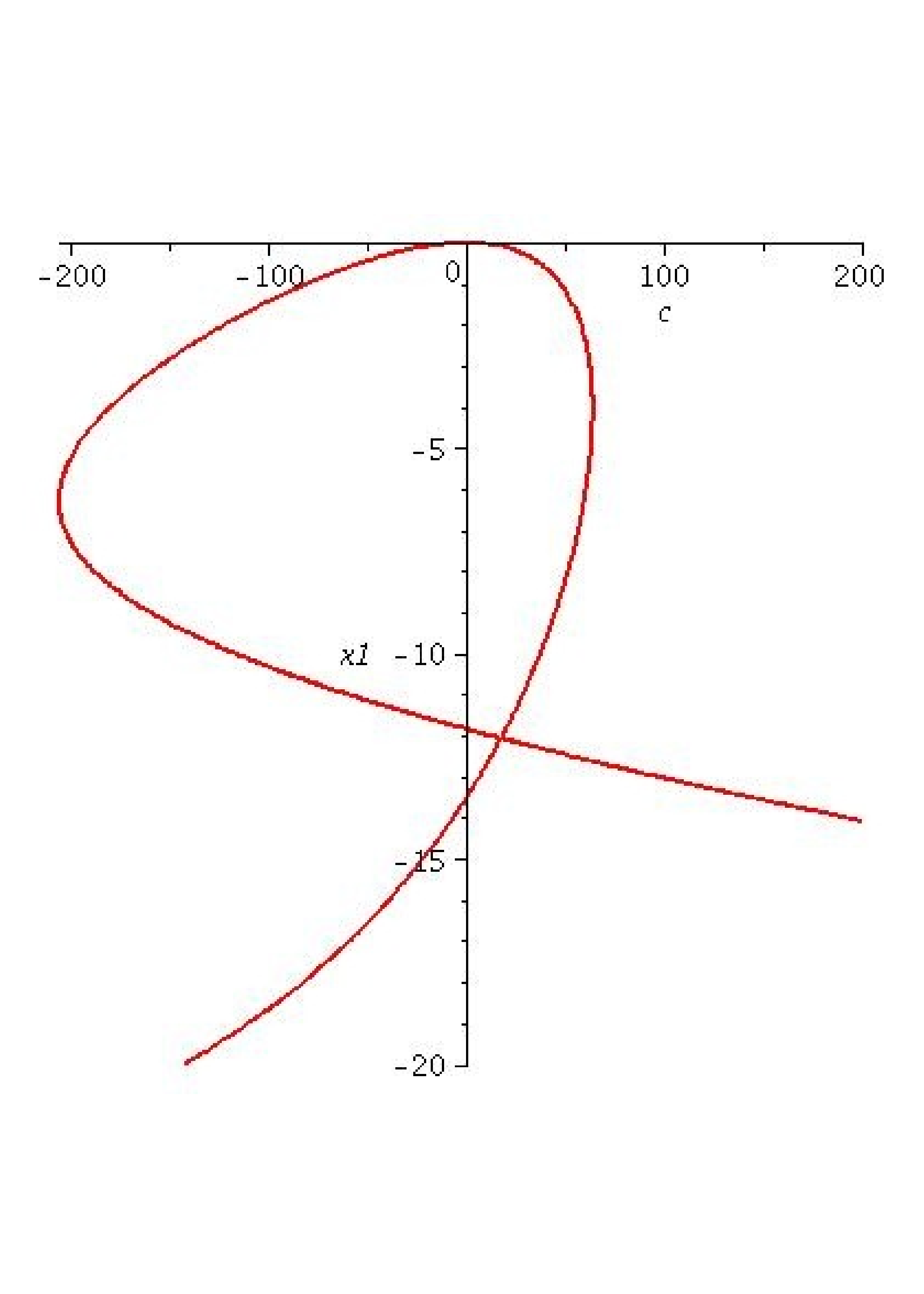}
 \end{tabular}
\caption{The algebraic curve $\Cp$ in different quadrants.}
\label{fig1234}
\end{figure}
Note that for small values $c^*$ there are four real
values of $x_1$ satisfying the degree $4$ polynomial equation $p(c^*,x_1)=0$ and
in a certain range, approximately for $c^*$ in the interval
$(44.43,58.33)$, there are $3$ positive solutions. In fact, the system shows
multistationarity in this range, and the middle values correspond to unstable
steady states.
Figure~\ref{fig:tang1} below presents the differences and similarities between the 
approximate plot of the implicit curve $\Cp$ 
 and the curve featuring hysteresis obtained via numerical ode simulation with \citeauthor{matlab}.
So, the black curve in Figure~\ref{fig:tang1} (a) approximates all the positive real
zeros $(c,x_1)$ of $p$. On the other side,
the curves (b), (c), (d)  are produced as approximate limit values 
via numerical integration of the ode system 
at different initial values. In Figure~\ref{fig:tang1} (c) and (d), the two curves in (b) are 
depicted separately. The initial values for the curve in (c) vary with $c$ and are 
$x_{10}^*=${\small [MAPKK]}$=c$, $x_3^*=${\small[M]}$=500$, $x_{11}^*=${\small [MKP3]}$=100$, 
and the other variables are set to zero. The value of $x_1$ represented is (approximately) the 
equilibrium value (computed for a big enough value of time $t$).
The curve in (d) should be read ``backwards'', starting at $c=100$ with the corresponding equilibrium 
point of the curve in (c) as initial state. At each step, the total amount of {\small [MAPKK]} is reduced 
from the previous equilibrium, keeping the same stoichiometric compatibility class for each value of $c$, 
until $c=0$ is reached. Two ``fake'' traces appear in this usual numerical picture: those that go from the lower
stable values of $x_1$ to the higher stable values of $x_1$ and back, which are produced
by an intent of the plotter to interpolate continuously the solutions of the simulations 
(there are also some inaccuracies due to the numeric approximation). (See the Supplementary Material we
provide for the \citeauthor{matlab} script used to produce Figure~\ref{fig:tang1} (b), (c) and (d).)
Note that the middle unstable steady state values in the black curve in (a) (in the multistationarity range)
are not taken as the initial values, and they do not occur in the blue and green curves in (b).

\begin{figure}[ht]
\begin{tabular}{cc}
 (a)\includegraphics[scale=.21,trim=2mm 44mm 2mm 4cm, clip=true]{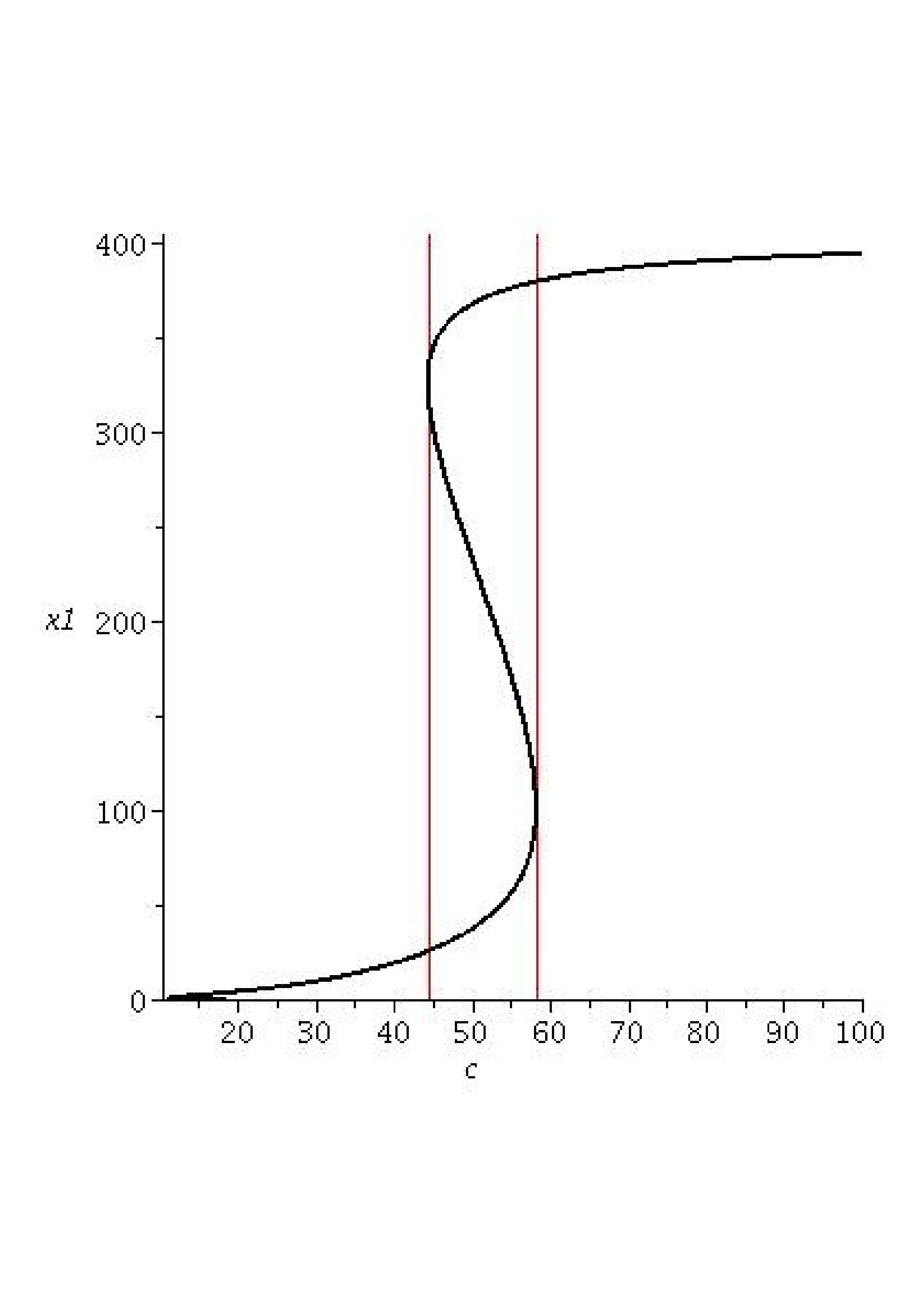} &  
 (b)\includegraphics[scale=.3,trim=2cm 7cm 2cm 7cm, clip=true]{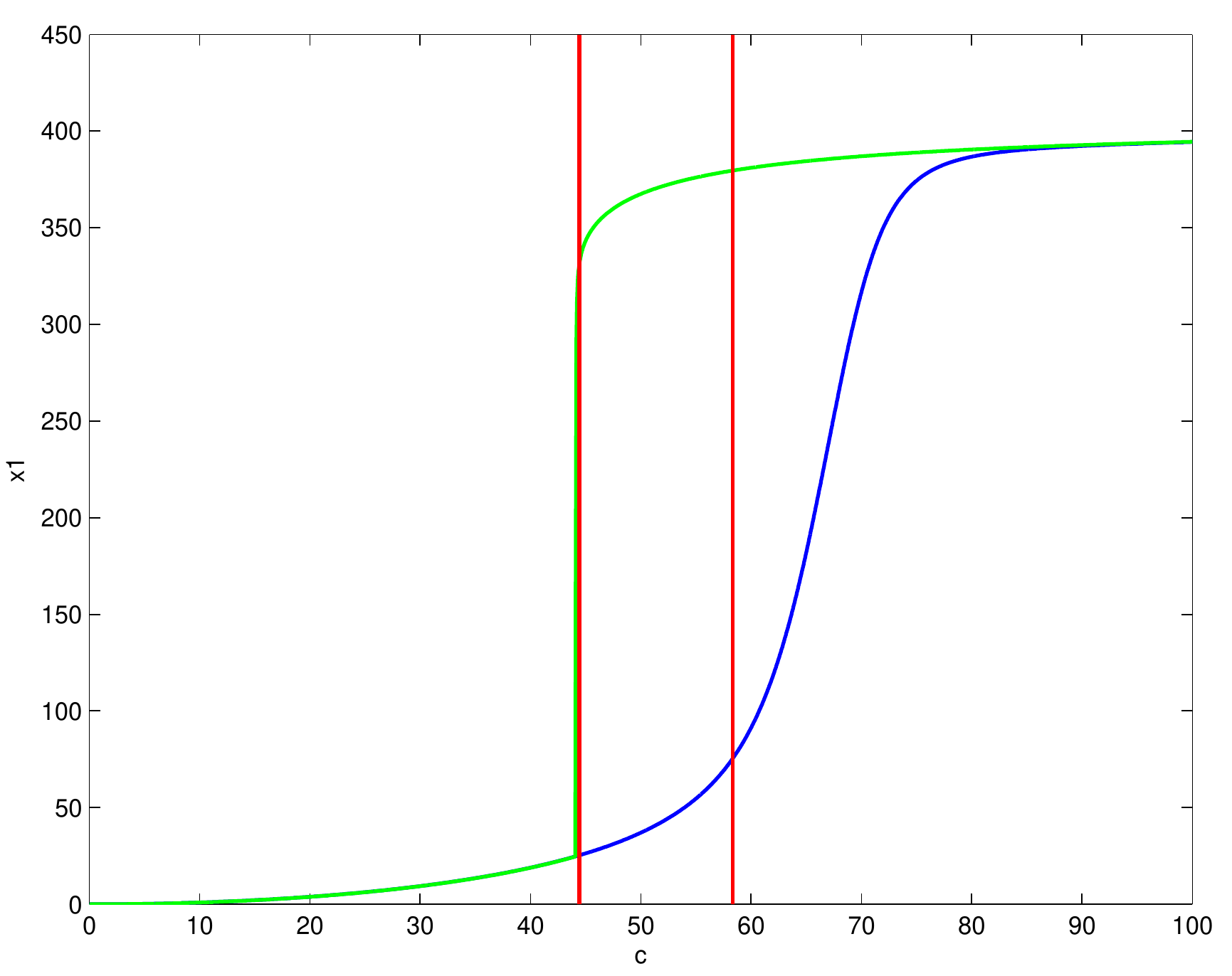} \\
 (c)\includegraphics[scale=.3,trim=2cm 7cm 2cm 7cm, clip=true]{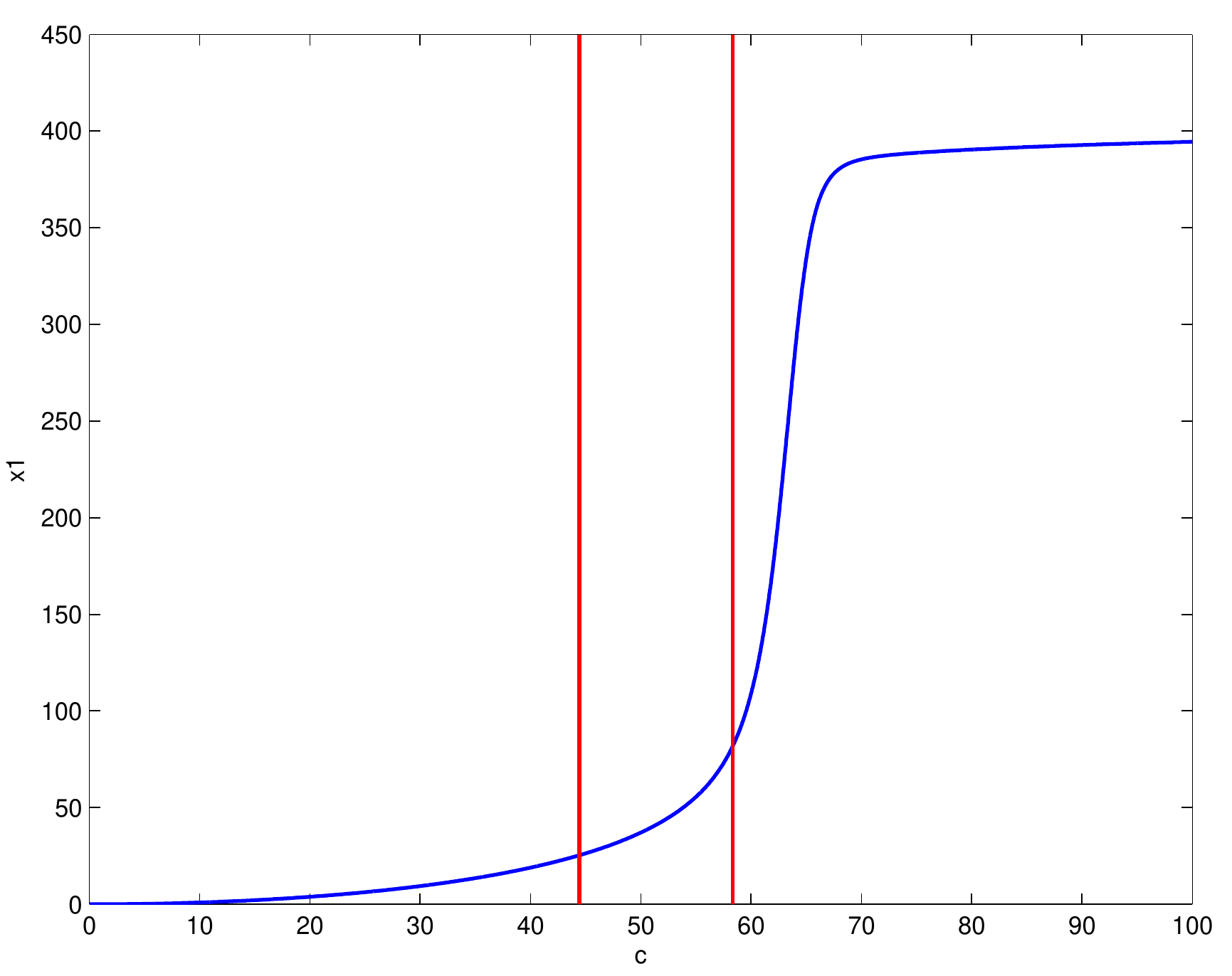} &  
 (d)\includegraphics[scale=.3,trim=2cm 7cm 2cm 7cm, clip=true]{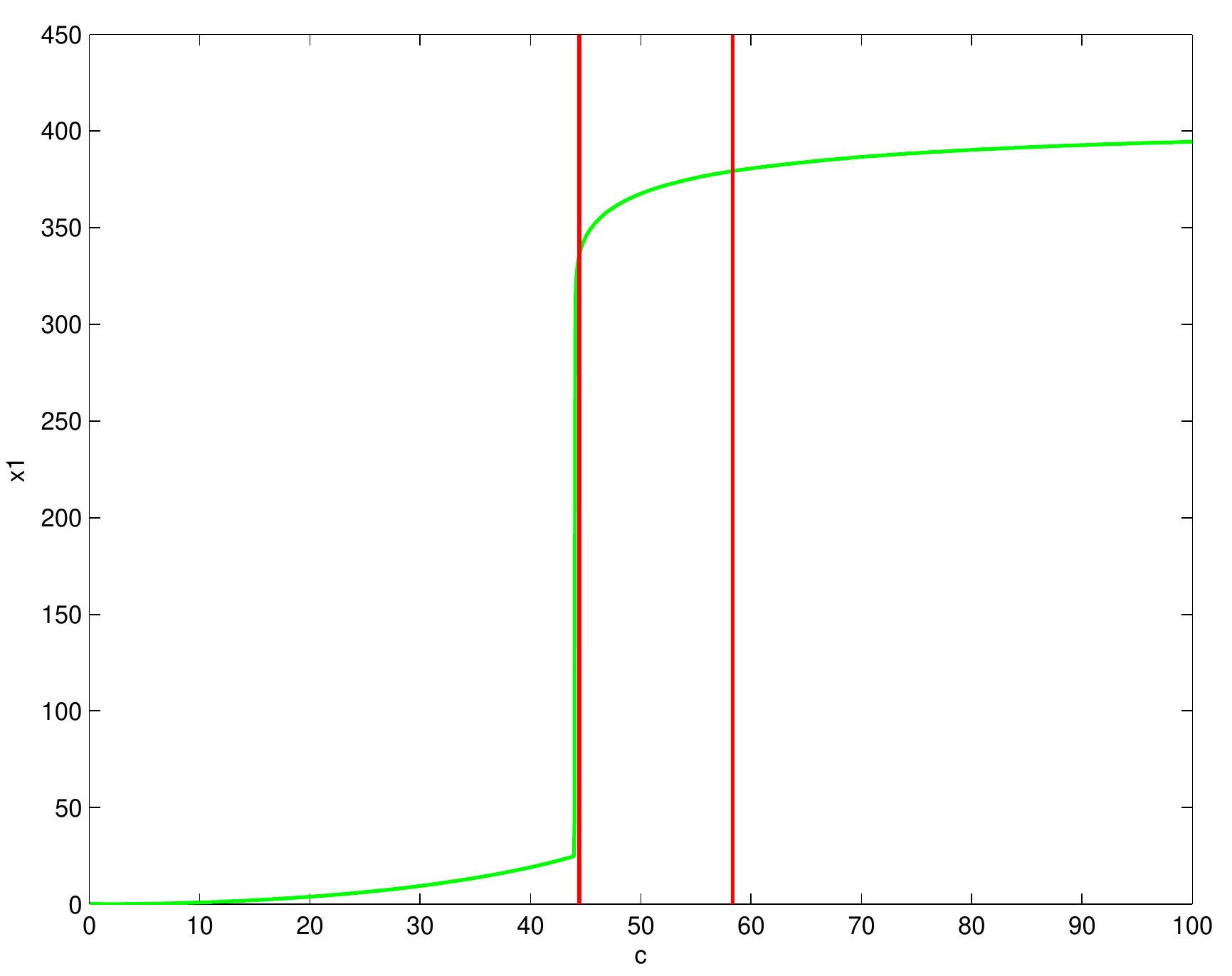}
\end{tabular}
\caption{The dose-response curves in the positive orthant with the vertical lines $c=44.43$ and $c=58.33$. 
(a): The plot of the implicit curve $\Cp$ in black with Maple. 
(b): The standard hysteresis simulation diagram with MATLAB, which is the superposition of the curves
in (c) and (d).
(c): The blue curve in (b) from the lower steady state values of $x_1$ to the higher steady state values. 
(d): The green curve in (b) from the higher steady state values of $x_1$ to the lower steady state values.}
\label{fig:tang1}
\end{figure}

\subsection{Taking $x_2$ as output variable}\label{ssec:x2c}

If we now follow the same procedure but eliminating all variables except for $x_2$ and $c$, we obtain a polynomial 
$q(x_2,c)$, again with degree $4$ in the variable $c$. 
The resultant $R_{4}$ of $q$ and 
$\frac{\partial q}{\partial c}$, thought of as polynomials in $\mathbb{R}[x_2][c]$, is a 
polynomial in $\mathbb{R}[x_2]$ of  degree $38$
with only six positive real roots.
The values of these positive roots are approximately 
$\alpha_1 \approx 15.2$,  $\alpha_2\approx 8.49$, $\alpha_3\approx 3.02$,
$\alpha_4\approx 2.1347$, $\alpha_5 \approx 2.1345$, and $\alpha_6\approx 2.1263$.

As $q(20,c)$ has no real roots, $\alpha_1$ is a nontrivial upper bound for 
$x_2$ at steady state by the first part of Theorem~\ref{thm:main}. To use the second 
part of this theorem, we must focus on the roots of the leading coefficient in $c$, 
which has no positive root. Hence, the only nontrivial bound we can find is $\alpha_1\approx 15.2$.

The corresponding implicit dose-response curve $\Cp'=\{(c,x_2) | q(c,x_2)=0\}$ has the unexpected
shape featured in Figure~\ref{fig:tang2}(a). For $c^*$ in the same approximate range $(44.43,58.33)$,
there are more than one positive solutions $x_2$ to the polynomial equation $q(c^*,x_2)=0$. The higher
values correspond to unstable steady states. The lower values values cannot be completed to a
nonnegative steady state, as we now explain with the help of Figure~\ref{fig1250}. We remark that
those points do not lie on a line, as the approximate picture seems to show.
\begin{figure}[ht]
\centering
\includegraphics[scale=.25,trim=2mm 46mm 2mm 46mm, clip=true]{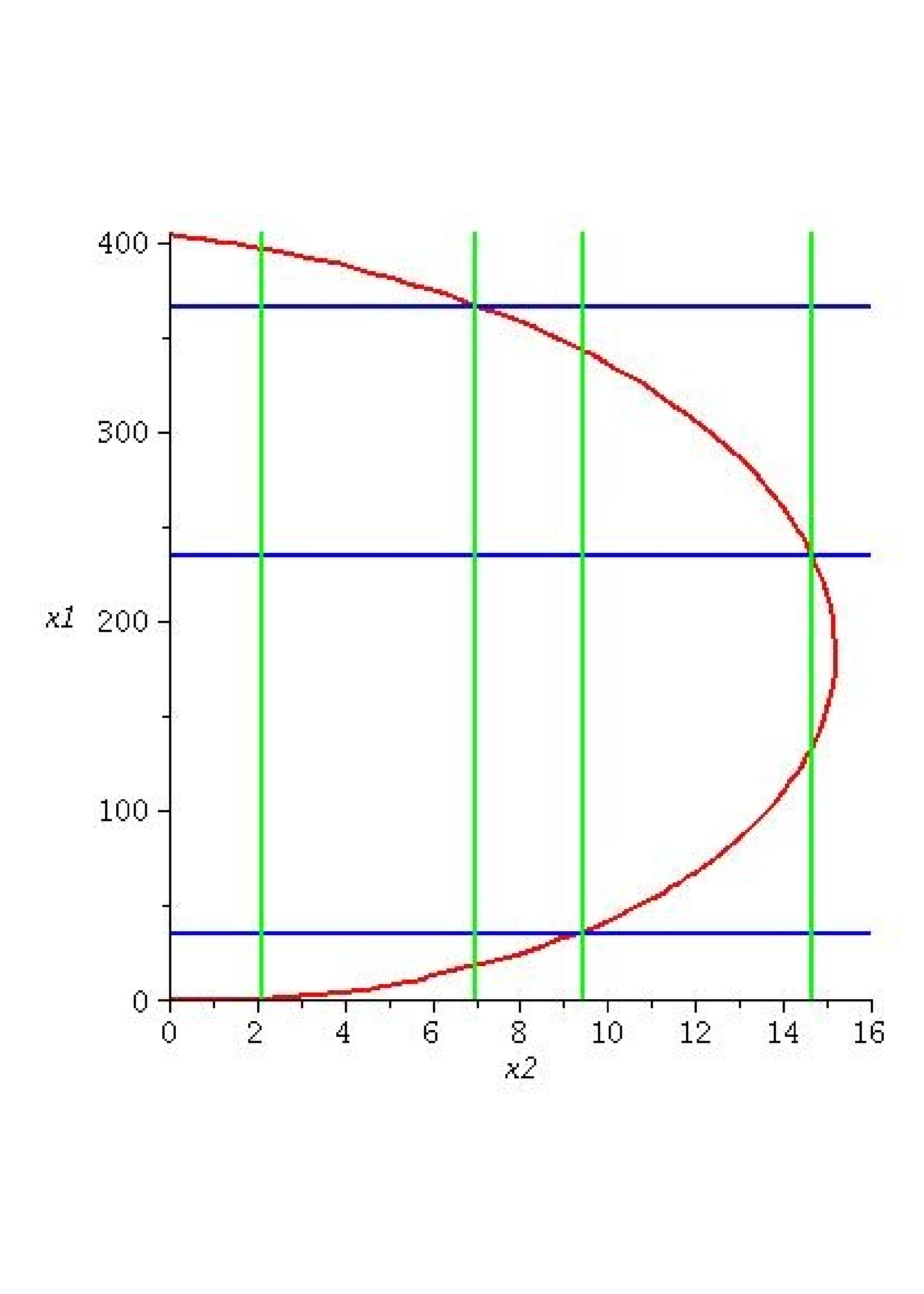}
\caption{The curve $g(x_1,x_2)=0$ (red), the horizontal lines that represent $p(50,x_1)=0$ 
and the vertical lines that represent $q(50,x_2)=0$, in the positive orthant}
\label{fig1250}
\end{figure}
We can instead eliminate from $I$ all variables (including $c$) but $x_1$ and $x_2$. We computed
with Singular a nonzero polynomial $g=g(x_1, x_2) \in I$, which relates the values of $x_1$
and $x_2$ at steady state for any value of $c$.  We pick the value $c^*=50$ in the
multistationarity range. Figure~\ref{fig1250} depicts:
\begin{itemize}
 \item[$\bullet$] the (approximate) curve $\{g=0\}$,
 \item[$\bullet$] the horizontal lines defined by $p(50,x_1)=0$,
 \item[$\bullet$] the vertical lines defined by $q(50,x_2)=0$,
\end{itemize}
in the nonnegative orthant of the $(x_2,x_1)$-plane. For any steady state for which $c=50$,
the values of its second and first coordinates have to be in the intersection of these
three pictures. We see that no point with $x_2$ close to $2$ satisfies this property.

One can now guess the shape of the hysteresis simulation diagram (produced with MATLAB), which is
shown in Figure~\ref{fig:tang2}(b). Again, it is interesting to detect the ``fake'' traces
in the plot, when comparing with the implicit dose-response curve $\Cp'$ in Figure~\ref{fig:tang2}(a).
\begin{figure}[ht]
\begin{tabular}{cc}
 (a)\includegraphics[scale=.2,trim=2mm 46mm 2mm 46mm, clip=true]{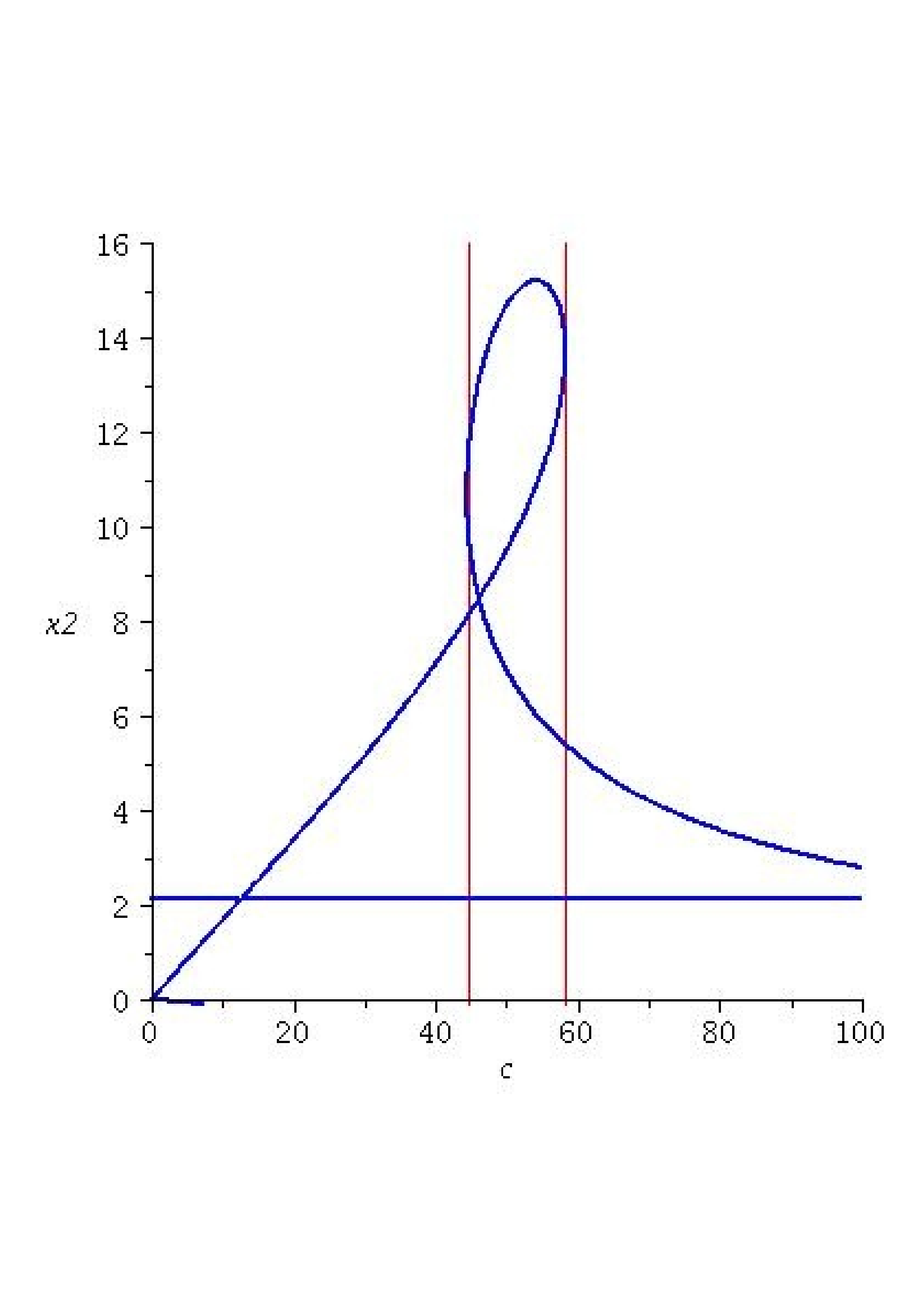}& 
 (b)\includegraphics[scale=.3,trim=2cm 7cm 2cm 7cm, clip=true]{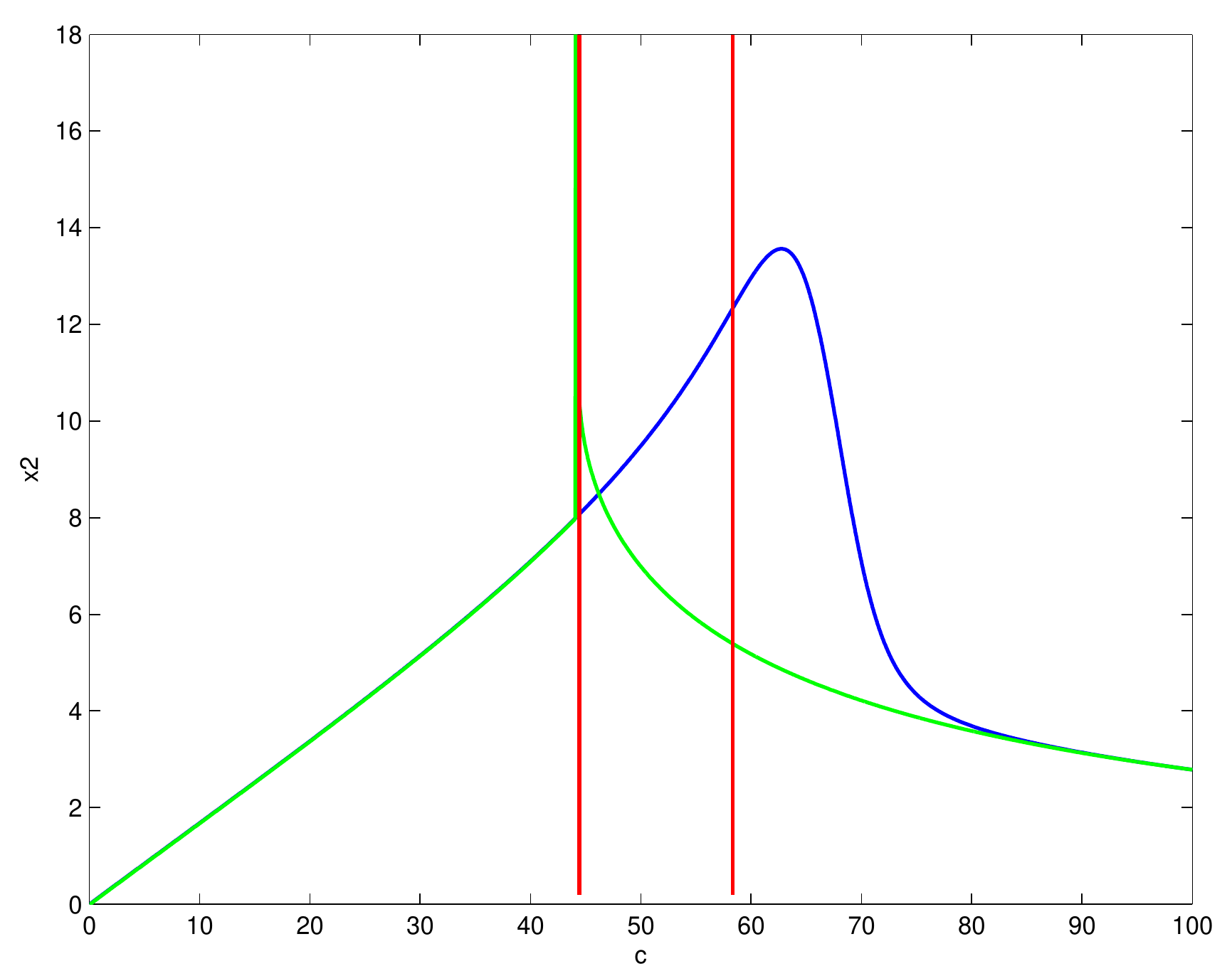}
\end{tabular}
\caption{The curve $q(c,x_2)=0$ in the positive orthant with the lines $c=44.43$ and $c=58.33$ in red. 
(a): The plot of the implicit curve with Maple. 
(b): Simulation with MATLAB. }
\label{fig:tang2}
\end{figure}

\subsection{Moving the parameters}\label{ssec:moving}

As we remarked in Section~\ref{sec:results}, by a variant of Lemma~\ref{lem:p}, one could get 
a polynomial $p$ in the ideal $I$ depending on some of the parameters $b_2, \dots, b_r$ or some rate 
constants. In theory, this is possible. In practice, even if one could compute $p$ effectively,
the output might be too big to be understood. For our running example, we also consider as
variables (besides
the $x_i$'s) the following: $c=b_1$ (MAPKK$_{tot}$), 
  $b_2 $ (M$_{tot}$), $b_3$ (MKP3$_{tot}$), 
and $h_1$. 

In this case, we can compute a polynomial
$p=p(c, b_2, b_3, h_1, x_1)$ in the parametric steady state ideal 
(considered in $\R[c,b_2,b_3,h_1,x_1,\dots,x_{11}]$) 
of total degree $13$, degree $6$ in $x_1$ and degree $4$ in $c$. 
The coefficient $p_4(b_2, b_3, h_1, x_1)$ of $c^4$ in $p$ equals:  
{\small                                                                                                
\[( \eta_1+\eta_2 x_1 h_1+ \eta_3 x_1^2 h_1^2)(4095 x_1^2 h_1+4118 x_1 h_1 b_3-4095 x_1 h_1 b_2+4095 x_1-4095 b_2),\]}
with {\small $\eta_1=55891160325>0$},  {\small $\eta_2=93839717790>0$}, and {\small $\eta_3=80925216157>0$}. 
The biggest positive asymptote $x_1 = \alpha$ is given by the only positive root $\alpha$ of the right factor, which is
easily seen to be smaller than the trivial bound $b_2$. 
Moreover, the difference $b_2 - \alpha$ is the following positive quantity:
\begin{equation}\label{eq:b2}
 b_2 - \alpha= \dfrac{\mu -\sqrt{\mu^2-4 \nu h_1^2 b_2 b_3}} {2 h_1},
\end{equation}
where $\mu = 1 + h_1 (b_2 +\nu b_3)$ and $\nu = \frac{4118}{4095}$.  
For fixed $b_2$ and $b_3$,  we can see for instance that $\alpha$ tends to the trivial bound
$b_2=\text{M}_{tot}$ when the dephosphorylation reaction constant  $h_1$ tends to zero, as expected.
In general, one can try to perform the computations keeping a few relevant parameters as variables,
to get a precise implicit description of the dependence of the steady state values on these parameters.

\section{DISCUSSION} 

We have introduced a novel approach for the study of dose-response curves, that is, for the
relation between steady state coordinates and input variables in autonomous polynomial dynamical
systems, via implicit curves. This analysis is possible regardless the absence
of explicit expressions or the presence of multistationarity and gives explicitly the implicit
relations between the input and the output.

As an application, we made a thorough study of one of the enzymatic mechanisms in~\cite{markevich},
where we obtained nontrivial bounds at steady state. We 
also used this example to point out how to understand the usual pictures featuring hysteresis and
to show that the implicit curves can be too difficult to be obtained ``by hand'' but they
can nevertheless be used to extract interesting conclusions on the behavior of the system.

\section*{Appendix}\label{ap:A}

We include the proofs of Lemma~\ref{lem:p}, Lemma~\ref{lem:IFT} and Theorem~\ref{thm:main}.

\begin{proof}[Proof of Lemma~\ref{lem:p}]
The proof uses basic results on the dimension of algebraic varieties, that can be found for instance
in \cite[Chapter~1,\S~6]{shafa}. As we remarked after the statement of Lemma~\ref{lem:p}, the ideal
$I$ can be generated by $s$ polynomials in $s+1$ variables, and so its dimension $d$ is at least $1$.
Consider the projection map $\pi(c,x_1,\dots,x_s)=c$ from the variety $V(I)$ of zeros of $I$
in $\C^{s+1}$. If there exists a nonzero polynomial $q \in I \cap \R[c]$,
the image of $V(I)$ would be contained in the zero set of $q$ in $\C$ and would be therefore of
dimension $0$.
By hypothesis, the fiber over any of those points has also dimension $0$ but then from the
Fibre Dimension Theorem we would get $0 \ge d-0 \ge 1$, a contradiction. Therefore,
the image is dense in $\C$ and so its closure has dimension $1$. Again, by the same Theorem,
we deduce that the dimension of $I$ equals $1$.
This implies that, given $2$ (or more) 
variables, it is possible to find a nonzero polynomial in those variables in the ideal. In particular, 
we can find a nonzero polynomial $p=p(c,x_1)$ in $I\cap\R[c,x_1]$ with positive degree in $x_1$.
\end{proof}

\begin{proof}[Proof of Lemma~\ref{lem:IFT}]
 It is enough to show that for any 
$\beta \in \Omega$ there exists a neighborhood where the cardinality is constant.
By Lemma~\ref{lem:res}, we have that $p_n(\beta) \neq 0$ and
 $\frac{\partial p}{\partial c}(c,\beta)\neq 0$ for all $c$ such that $p(c,\beta)=0$.
Let $d = \#\Cp_{\beta}$ and
call $\Cp_{\beta}=\{c^1,\dots,c^d\}$.
Using the IFT, it is possible to find an open set $V$ around $\beta$ and, 
for $1\leq i\leq d$, open sets  $U_i$ around each $c^i$ with 
 $U_i\cap U_j=\emptyset$ if $i\neq j$, 
and smooth functions $g_i:V\to U_i$ in $\mathcal{C}^1$ such that
$$\{(g_i(x_1),x_1)~|~ x_1\in V\}=\{(c,x_1) \in U_i\times V ~|~ p(c,x_1)=0\}.$$
Thus, for all $x_1 \in V$ we have that $\#\Cp_{x_1} \ge d$.

 Suppose there exists a sequence $\gamma_m \to \beta$ for $m \to \infty$ with $\gamma_m\in V$ and
 $\#\Cp_{\gamma_m} > d$. For every $m$, choose a point $c_m$ with $p(c_m, \gamma_m)=0$
 and $c_m \neq g_i(\gamma_m)$ for all $i=1, \dots,d$.
 Since $x_1=\beta$ is not an asymptote, the sequence $(c_m)$ is bounded and thus there
 exists a convergent subsequence of $((c_m, \gamma_m))$ which converges to a point $(c^*,\beta)$.
 But then $p(c^*, \beta)=0$ and $c^*$ is different from $c^1, \dots,c^d$, a contradiction.
 \end{proof}

\begin{proof}[Proof of Theorem~\ref{thm:main}]

 Let $\alpha_1> \alpha_2 > \dots> \alpha_m$ be the real zeros of $R_{n}$, 
 and $k \in \{1,\dots,m\}$, $\beta_1, \dots, \beta_k$ be as in the hypotheses 
 of Theorem~\ref{thm:main}. The open intervals $(\alpha_i,\alpha_{i-1})$ for 
 $1\leq i\leq k$ and $\alpha_0:=+\infty$ are 
 connected components of the complement of the zeros of $R_n$. By Lemma~\ref{lem:IFT}, 
 $\#\Cp_{x_1}=\#\Cp_{\beta_i}=0$ for all $x_1 \in (\alpha_i,\alpha_{i-1})$ for $1\leq i \leq k$.
 There are no zeros of $p$ with $x_1=\alpha_i$ ($1\leq i\leq k$) because 
 $\Cp_{\alpha_i} \cap J = \emptyset$ by hypothesis,
 which means that there are no nonnegative steady states with $x_1=\alpha_i$. 
 Therefore, we have $x_1< \alpha_k$ at any nonnegative steady state. This is, $\alpha_k$ is an upper bound 
 for $x_1$ at steady state.
 
 Let $\alpha$ be as in the second part of Theorem~\ref{thm:main}. Denote by $X$ the set 
 $X:=\{x_1>\alpha : \textrm{there exists}\; c \in J \; \text{with} \; p(c,x_1)=0\}$. 
 By the first part of the theorem, we have $x_1<  \alpha_k$ for all $x_1$ in $X$. 
 Suppose $X\neq \emptyset$, and let us call $\mu$ the supremum of $X$.
 If $x_1=\mu$ were an asymptote, we would have $p_n(\mu)=0$, which is not possible because $\mu>\alpha$. 
 Then, for any sequence $((c^{(m)},x_1^{(m)}))\subset \Cp$ with $x_1^{(m)}\in X$, $x_1^{(m)}\to \mu$, 
 and $c^{(m)}\in J$, there exists a convergent subsequence such that the first coordinates tend to 
 some $c^*\in J$ ($J$ is closed). Since $p$ is continuous, 
 we have $p(c^*,\mu)=0$, and by hypothesis, as $\alpha<\mu\leq\alpha_k$, $\frac{\partial p}{\partial c}(\mu,c^*)\neq 0$. 
 Then, by the IFT, there exist $\delta>0$ and a smooth function $g$ such that $p(g(x_1),x_1)=0$ for all 
 $x_1\in(\mu,\mu+\delta)$. If $J=\R$, this is impossible because $\mu$ is the supremum. If $J=[0,+\infty)$, then 
 $\mu$ is a maximum and $c^*=0$, but this is not possible by hypothesis, since $p(0,x_1)\neq 0$ for all 
 $x_1>\alpha$. Therefore, $X=\emptyset$ and $x_1\leq \alpha$ at every nonnegative steady state.
 
\end{proof}


\begin{thebibliography}{}


\bibitem[Cox et~al.(2007)]{IVA}
Cox D., Little J., O'Shea D.  (2007),
\newblock Ideals, varieties and algorithms.
Undergraduate Texts in Mathematics, Third Edition, Springer, New York.


\bibitem[Singular(0000)]{singular}
Decker W., Greuel G.-M., Pfister G., Sch{\"o}nemann H. (2012),
\newblock {\sc Singular} {3-1-6} --- {A} computer algebra system for polynomial computations.
\newblock {http://www.singular.uni-kl.de} 

\bibitem[Feinberg(1979)]{feinberg}
Feinberg M. (1979),
\newblock {Lectures On Chemical Reaction Networks}.
\newblock Ohio State University.
\newblock {http://www.crnt.osu.edu/LecturesOnReactionNetworks}

\bibitem[Feinberg and Horn(1977)]{feho77}
Feinberg M., Horn F. (1977),
\newblock Chemical mechanism structure and the coincidence of the stoichiometric and kinetic subspaces.
\newblock Arch. Ration. Mech. Anal. 66(1), pp. 83--97.

\bibitem[Feliu and Wiuf(2012)]{fw12}
Feliu E., Wiuf C. (2012),  
\newblock Enzyme-sharing as a cause of multi-stationarity in signalling systems.
\newblock J. R. Soc. Interface  9(71), pp. 1224--1232.
 
\bibitem[Feliu et~al.(2012)]{fkaw12}
Feliu E., Knudsen M., Andersen L., Wiuf C. (2012), 
\newblock An algebraic approach to signaling cascades with $n$ layers.
\newblock Bull. Math. Biol. 74(1),  pp. 45--72.

\bibitem[Flockerzi et~al.(2013)]{fhc13}
Flockerzi D., Holstein K., Conradi C. (2013),
\newblock N-site phosphorylation systems with 2N-1 steady states.
\newblock Available at arXiv:1312.4774.
 
\bibitem[Gelf{$'$}and et~al.(1994)]{GKZ} 
Gelf{\cprime}and I.,Kapranov M., Zelevinsky A. (1994),
\newblock Discriminants, Resultants and Multidimensional Determinants.
\newblock Birkh\"auser, Boston. 

\bibitem[Horn and Jackson(1972)]{hj72}
Horn F., Jackson R.  (1972),
\newblock General mass action kinetics.
\newblock {{A}rch. {R}ation. {M}ech. {A}nal.}, 47(2), pp. 81--116.

\bibitem[Karp et~al.(2012)]{kpmddg}
Karp R., P\'{e}rez Mill\'{a}n M., Dasgupta T., Dickenstein A., Gunawardena J.  (2012),
\newblock Complex-linear invariants of biochemical networks.
\newblock J. Theor. Biol. 311, pp. 130--138.

\bibitem[Maple(0000)]{maple}
Maple 17 (2013),
\newblock Maplesoft, a division of Waterloo Maple Inc., Waterloo, Ontario.

\bibitem[Markevich et~al.(2004)]{markevich}
Markevich N., Hoek J., Kholodenko B. (2004),
\newblock Signaling switches and bistability arising from multisite phosphorylation in protein kinase cascades.
\newblock J. Cell Biol. 164(3), pp. 353--359.

\bibitem[MATLAB(0000)]{matlab}
MATLAB (2014),
\newblock version 8.3.0.
\newblock Natick, Massachusetts: The MathWorks Inc.

\bibitem[P\'erez Mill\'an et~al.(2012)]{pdsc12}
P\'{e}rez Mill\'{a}n M., Dickenstein A., Shiu A., Conradi C. (2012),
\newblock Chemical reaction systems with toric steady states.
\newblock Bull. Math. Biol. 74(5), pp. 1027--1065.

\bibitem[Shafarevich(1994)]{shafa}
Shafarevich, I. (1994),
\newblock Basic algebraic geometry. 1.
Varieties in projective space. Second edition.
\newblock Springer-Verlag, Berlin.

\bibitem[Tobis(2005)]{tobis}
Tobis E.~A. (2005),
\newblock Libraries for Counting Real Roots,
\newblock Reports on Computer Algebra (ZCA, University of Kaiserslautern) 34.

\bibitem[Vol{\cprime}pert and Hudjaev(1985)]{vh85}
Vol{\cprime}pert A.~I., Hudjaev S.~I. (1985),
\newblock Analysis in classes of discontinuous functions and equations of
  mathematical physics.
\newblock volume~8 of Mechanics: Analysis. Martinus Nijhoff Publishers, Dordrecht.

\bibitem[Wang and Sontag(2008)]{ws08}
Wang L., Sontag E.  (2008),
\newblock On the number of steady states in a multiple futile cycle.
\newblock J. Math. Biol. 57(1), pp. 29--52.



\end{thebibliography}
\end{document}